\documentclass[conference,twocolumn]{IEEEtran}
\usepackage{mathtools,dsfont}
\usepackage{psfrag,bbm,graphicx}
\usepackage{comment,balance}
\usepackage{epstopdf}
\usepackage{epsfig,subfig}
\usepackage{caption}
\usepackage{multicol}
\usepackage{multirow, cite}
\usepackage{amsfonts, color}
\usepackage{array,cuted,lipsum}
\usepackage{amssymb,amsmath,amsthm}
\usepackage[linesnumbered,ruled,vlined]{algorithm2e}
\usepackage[belowskip=-8pt,aboveskip=0pt]{caption}
\allowdisplaybreaks
\usepackage{tikz}

\newcommand{\PP}{\mathbb{P}}
\newcommand{\E}{\mathbb{E}}

\newcommand{\OO}{\text{O}}

\newtheorem{lemma}{Lemma}

\newtheorem{example}{Example}

\newtheorem{theorem}{Theorem}
\newtheorem{definition}{Definition}

\title{Correlated Age-of-Information Bandits}
\author{Ishank Juneja, Santosh Fatale, and Sharayu Moharir \\ Department
		of Electrical Engineering, Indian Institute of Technology Bombay} 
\newfont{\mycrnotice}{ptmr8t at 7pt}
\newfont{\myconfname}{ptmri8t at 7pt}

\begin{document}
\maketitle
\begin{abstract}
We consider a system composed of a sensor node tracking a time varying quantity. In every discretized time slot, the node attempts to send an update to a central monitoring station through one of $K$ communication channels. We consider the setting where channel realizations are correlated across channels. This is motivated by mmWave based 5G systems where line-of-sight which is critical for successful communication is common across all frequency channels while the effect of other factors like humidity is frequency dependent.

The metric of interest is the Age-of-Information (AoI) which is a measure of the freshness of the data available at the monitoring station. In the setting where channel statistics are unknown but stationary across time and correlated across channels, the algorithmic challenge is to determine which channel to use in each time-slot for communication. We model the problem as a Multi-Armed bandit (MAB) with channels as arms. We characterize the fundamental limits on the performance of any policy. In addition, via analysis and simulations, we characterize the performance of  variants of the UCB and Thompson Sampling policies that exploit correlation.  
\end{abstract}

\section{Introduction}
\label{sec:intro}
Future communication technologies including 5G are likely to use the millimeter band (30GHz to 300GHz) for communication. The available bandwidth is partitioned into frequency channels for communication. Factors such as frequency dependent atmospheric attenuation affect propagation in the millimeter band. In addition, the availability of a line-of-sight path between the source and receiver is critical for successful communication in this band \cite{doi:10.1002/9780470889886.ch1}. Since the existence of a line-of-sight path is frequency agnostic, channel realizations across different frequency channels at a given time are correlated. 

Age of Information (AoI), introduced in \cite{Kaul2012RealTime}, is a freshness of data metric that measures the time elapsed since the most recent successful update sent from a source was received at the intended destination. For time-critical applications like self-driving cars, smart homes and other up and coming IoT applications, it is imperative that the data used by the control unit to make decisions is as recent as possible. In these cases, AoI is a suitable performance metric.

The system we study builds on the setting studied in \cite{fatale2020regret} and consists of a single source node tracking a time-varying quantity. The source attempts communicating an update to a central monitoring station. At every discretized time step $t$, an update is sent through one among $K$ available channels. Each channel has a certain probability of success which is assumed to remain static across the period of operation. The channel statistics, that is the probability of success or failure of communication for a certain channel, are not known to the scheduler. However, it is known that the successes and the failures across the $K$ channels are correlated with one another through an underlying stochastic state $X$. Depending on the state in which $X$ finds itself, certain channels are successful whereas others are not (Figure \ref{fig:system_model}). Thereby the model accounts for correlation between the performances of channels. 



The algorithmic challenge is to determine which channel to use for communication in each time step in order to minimize cumulative AoI over a finite horizon $T$. The difference between expected cumulative AoI under the chosen scheduling policy, and under an oracle's optimal strategy of choosing the best channel $k^{*}$ at all time steps, is called \textit{AoI Regret}. We model the correlation between the performance of channels using the \textit{Correlated Multi-Armed Bandit} framework introduced by \cite{gupta2019multiarmed}. For the AoI metric, scheduling decisions taken at any time step have a downstream effect across all future time slots. Hence, a new analysis for the AoI regret metric is needed to tackle problem instances drawn from the Correlated Bandit framework.  
\begin{figure}[t]
\centering
\includegraphics[scale=0.46]{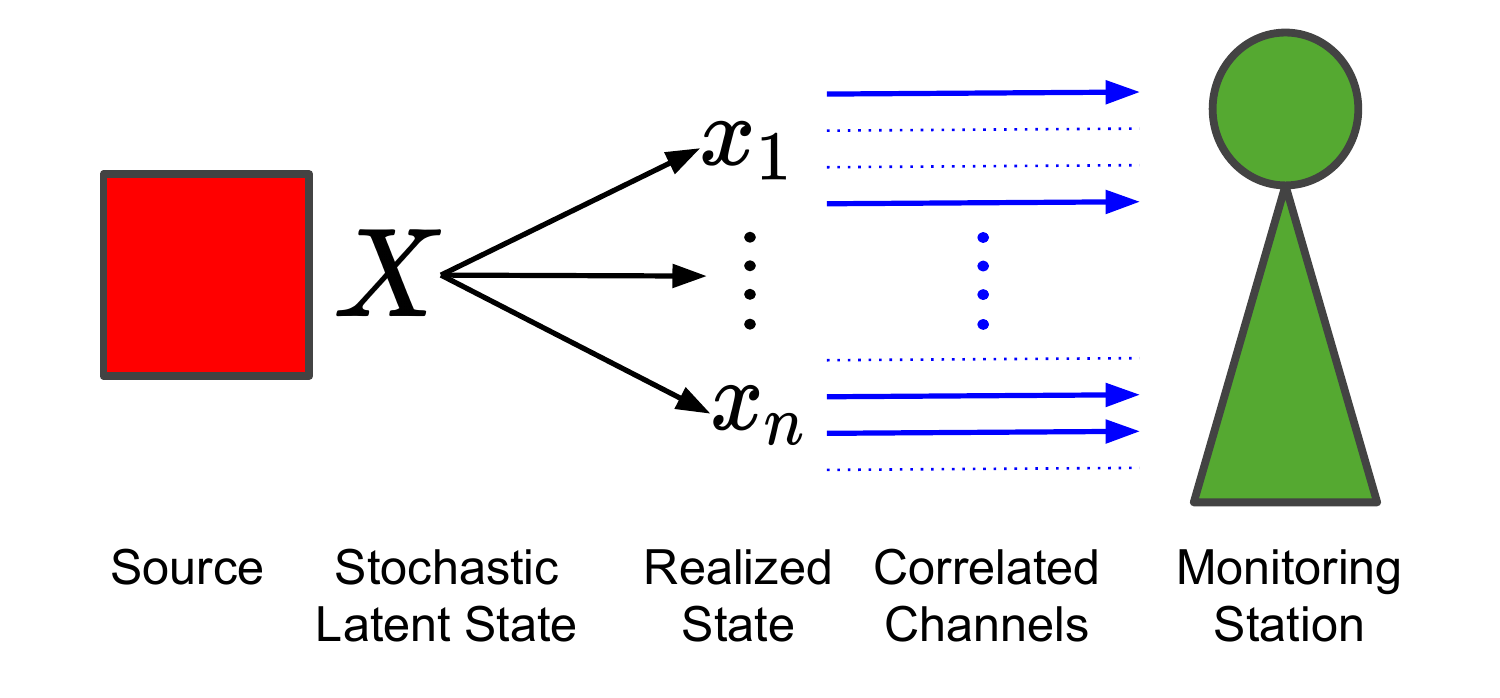}
\caption{The source node is attempting to communicate with the monitoring station. Depending on the state in which $X$ finds itself, only certain channels work.}
\label{fig:system_model}
\end{figure}
\subsection{Our Contributions}
\label{subsec:contrib}
\textit{Lower bound on AoI regret:} We show a lower bound of $\Omega(\log T)$ on AoI regret for instances that have at least one \textit{strictly competitive} arm (formally defined in Section \ref{sec:setting}).

\textit{Performance of variants of UCB and Thompson Sampling:} We show that the AoI regret for Correlated-UCB (CUCB) and Correlated-Thompson Sampling (CTS) (proposed in \cite{gupta2019multiarmed}) is $\OO(C \log T) + \OO(1)$. Here $C$ is the number of sub-optimal \textit{competitive arms} (formally defined in Section \ref{sec:setting}). 

\textit{Empirical validation:} Through simulations, we compare the performance of UCB, Thompson Sampling, CUCB, CTS, and their \textit{AoI-Aware} variants proposed in \cite{fatale2020regret}.

\subsection{Related Work}
Multi-Armed Bandits (MABs) are a sequential decision making framework where at every time step $t$, a choice has to be made between $K$ possible \textit{bandit arms} with unknown statistics. 
UCB \cite{lai1985asymptotically} and Thompson Sampling \cite{thompson1933likelihood} are two widely studied algorithms for the MAB problem. In this work, we study variants of these policies more suited for our setting.


Recently, variants of the traditional MAB framework that are capable of incorporating additional structure into the decision making problem have been introduced. In addition to observing rewards, the \textit{contexual bandit} framework \cite{zhou2015survey}, learns a mapping between a context vector $\theta$ and the best arm $k^{*}$. Another example is the structured bandit framework \cite{structured}, in which the mean rewards for all arms as a function of the context $\theta$ are known but $\theta$ itself is hidden. The Correlated Multi-Armed Bandit framework of \cite{gupta2019multiarmed} is a variant of the MAB problem that presents the scheduler with arms whose rewards are not independent of each another. That is, sampling arm $k$ can reveal information about the rewards we can expect from another arm $\ell$. In addition to observing the rewards obtained from sampling an arm $k$, scheduling algorithms cognizant of correlation can track additional side information to identify some arms as sub-optimal without sampling them often, thereby reducing the accumulation of regret.

AoI or Age-of-Information is measure of the freshness of data available at the central monitoring station. AoI has been the focus of a variety of work, and \cite{kosta2017age} can be referred to for a comprehensive survey. Previously, a large focus of the work on AoI has been on problems where channel statistics are known \cite{tripathi2017age, jhunjhunwala2018age, kadota2018optimizing, Sombabu:2018:AAS:3241539.3267734}. In our work, we take channel statistics to be unknown and apply Multi-Armed Bandits to the scheduling problem, which is the approach taken by \cite{fatale2020regret} for minimizing AoI regret. Building upon the work of \cite{fatale2020regret}, we drop the assumption of independence between channels and instead take that they follow the Correlated Multi-Armed Bandit model.

\color{black}

\section{Setting}
\label{sec:setting}
%
\subsection{Correlated Bandit Model}
\begin{figure}[t]
	\centerline{\includegraphics[scale=0.46]{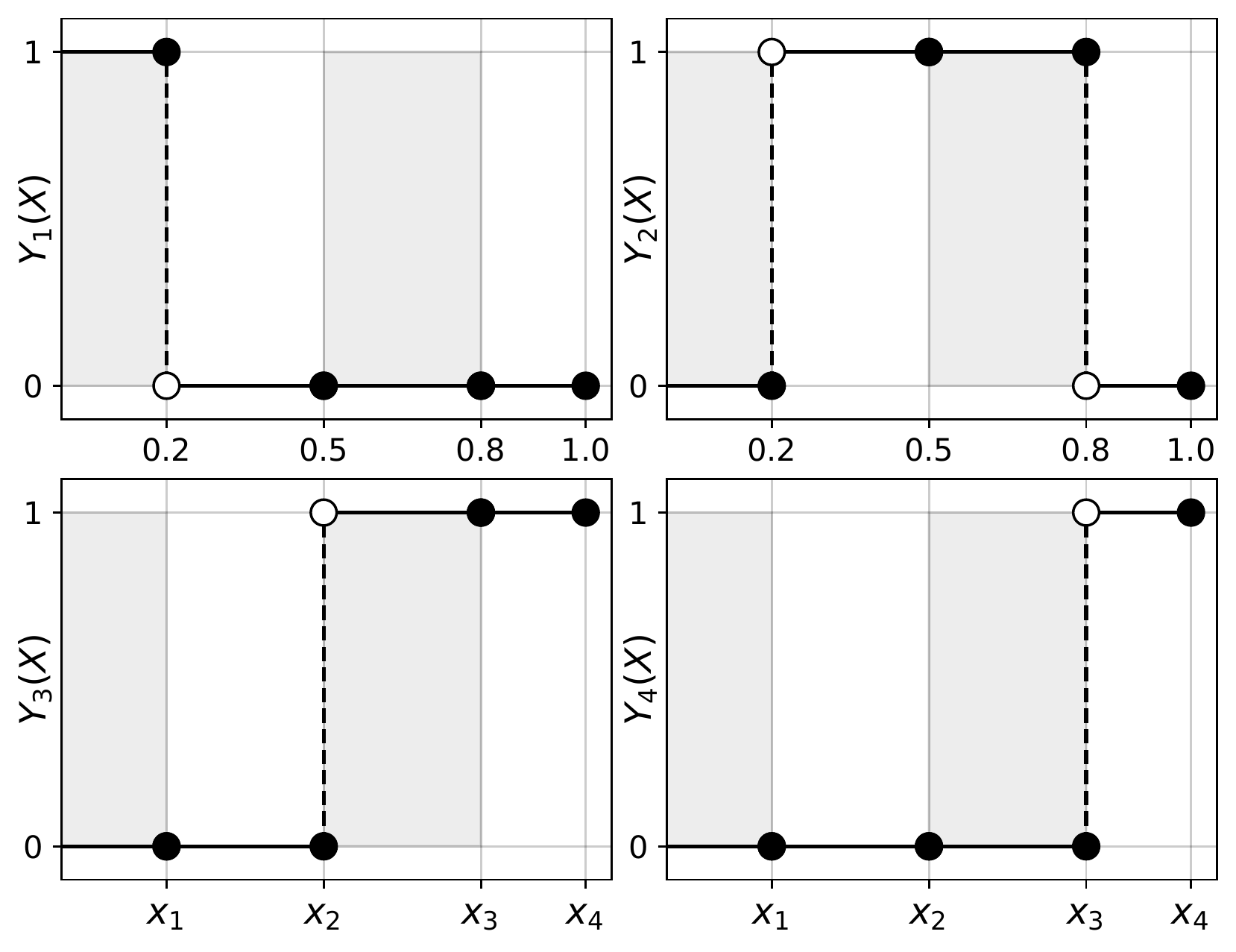}}
	\caption{Bandit instance $\mathrm{I}_1$} 
	\label{fig:instance_i1}
\end{figure}
\begin{figure}[t]
	\centerline{\includegraphics[scale=0.46]{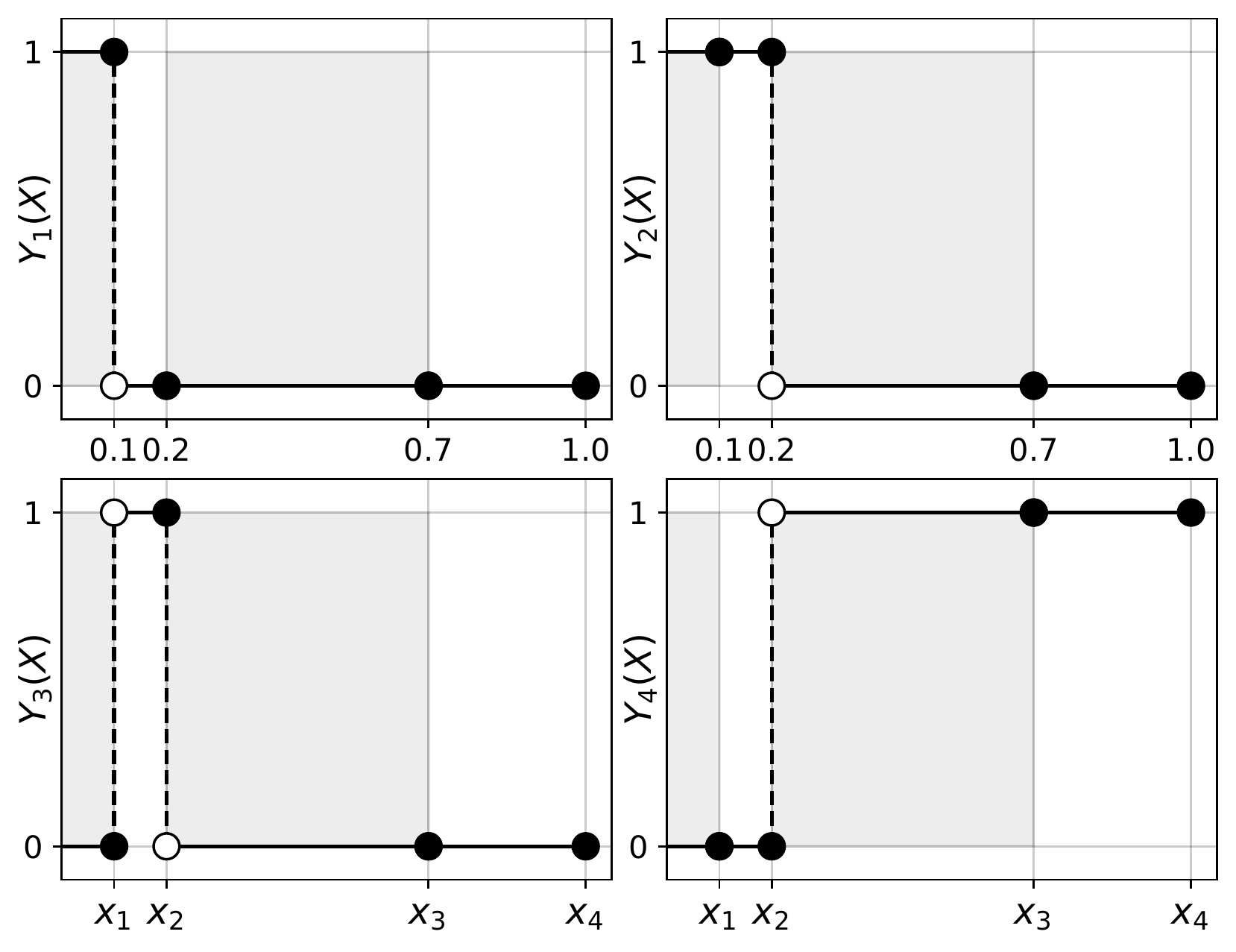}}
	\caption{Bandit instance $\mathrm{I}_2$} 
	\label{fig:instance_i2}
\end{figure}
For a system with $K$ communication channels we construct a Correlated Bandit instance with the same number of arms. The random variable $X$ captures the underlying state that determines the rewards for the correlated arms. The reward obtained on playing arm $k$ is denoted by a Bernoulli random variable $Y_k(X)$, where $Y_k$ is a known deterministic function. 
We define $\mu_k = \mathbb{E}_X [Y_k (X)]$. The optimal arm is denoted by $k^{*}$, and has mean $\mu^{*}$. The difference between the largest mean and the mean of a sub-optimal arm is called the sub-optimality gap and is given by $\Delta_k = \mu^{*} - \mu_{k}$.

Since the distribution of $X$ is unknown, the means $\mu_1, \mu_2, \ldots, \mu_K$ are also not known. MAB algorithms such as UCB empirically estimate the means $\hat{\mu}_k$ for use in decision making. Under the Correlated Bandit model, rewards for all the arms are functions of the same random variable $X$. Hence, we can infer the possible rewards that an arm $\ell$ could have returned if we had chosen arm $\ell$ instead of arm $k$. To exploit this, \cite{gupta2019multiarmed} introduces the notion of \textit{pseudo rewards}.  
\begin{definition}[Expected pseudo reward and pseudo gap]
	Pseudo reward for arm $\ell$ with respect to arm $k$ is given by,
	\begin{align}
		s_{\ell, k} (r) &= \sup_{x: Y_k(x) = r} Y_{\ell} (x).
	\end{align}
	Expected pseudo reward in turn is defined as,
	\begin{align}
	{\phi}_{\ell, k} &= \mathbb{E}_X [\,s_{\ell, k} (Y_{k} (X))\,].
	\end{align}
	The pseudo gap is defined as $\tilde{\Delta}_{\ell, k^{*}} = \mu^{*} - \phi_{\ell, k^{*}}$.
	\label{def:pseudo}
\end{definition}
We say arm $k$ is \textit{competitive} if $\tilde{\Delta}_{\ell, k^{*}} \leq 0$ and \textit{strictly competitive} if the inequality is strict. If $\tilde{\Delta}_{\ell, k^{*}} > 0$, we call arm $k$ non-competitive. We use $C$ to denote the number of competitive arms excluding arm $k^{*}$. 

\begin{example}
Two examples of Correlated Bandit instances $\mathrm{I}_1$ and $\mathrm{I}_2$ are shown in Figures \ref{fig:instance_i1} and \ref{fig:instance_i2} respectively. Both instances have $K=4$ arms and have $X$ as a discrete random variable taking values in the abstract alphabet $\{x_1, x_2, x_3, x_4\}$. In the figures, the length of the interval corresponding to $x_i$ is equal to $\mathbb{P}\{X = x_i\}$. Consider $\mathrm{I}_1$, in it, the mean reward for the bandit arms is given by, $\mu_1 = 1 \times 0.2$, $\mu_2 = 1\times0.3 + 1\times0.3$, $\mu_3 = 1\times0.3 + 1\times0.2$ and $\mu_4 = 1\times0.3$. Hence arm 2 is optimal with $\mu^{*} = 0.6$. The expected pseudo rewards can be computed using Definition \ref{def:pseudo} to obtain ${\phi}_{1, 2} = 0.4$, ${\phi}_{3, 2} = 1.0$ and ${\phi}_{4, 2} = 0.4$. Hence, for instance $\mathrm{I}_1$ arm 2 is optimal and arm 3 is the only competitive sub-optimal arm. Therefore $C=1$ for instance $\mathrm{I}_1$. Performing a similar analysis on $\mathrm{I}_2$, we find that arm 4 is optimal and no other arm is competitive. In other words $C=0$ for $\mathrm{I}_2$.  
\end{example}

\subsection{The AoI Regret Metric}
Here, the notions of Age-of-Information (AoI) and AoI regret, that were earlier described informally, are made precise.
\begin{definition}[Age-of-Information (AoI)]
At the start of time slot $t$, let $a(t)$ denote the AoI at the central monitoring station and let $u(t)$ denote the time index at which the recent most successful update was received by the monitoring station. Then,
$
	a(t) = t- u(t).
$
Alternatively, 
\begin{align*}
a(t) = 
\begin{cases}
1 & \text{if the update in $t-1$ succeeds} \\
a(t-1) + 1 & \text{otherwise.}
\end{cases}
\end{align*}    
\end{definition}	
Under a given scheduling policy $\rho$, let $a_{\rho}(t)$ denote the AoI in time slot $t$. Further, let $a^{*}(t)$ be the AoI under an oracle's policy that uses the optimal arm $k^{*}$ in all time slots. The AoI regret at time $T$ is the cumulative difference in expected AoI for the two policies from time-slots $1$ to $T$.  
\begin{definition}[Age-of-Information Regret (AoI Regret)]
AoI regret for a policy $\rho$, over $T$ slots is given by,
\begin{align}
R_{\rho}(T) = 
\sum_{t=1}^{T} \mathbb{E}[a_{\rho}(t) - a^{*}(t)] = \sum_{t=1}^{T} \mathbb{E}[a_{\rho}(t)] - \frac{T}{\mu^{*}},
\label{eq:aoi_regret_def}
\end{align}
where \eqref{eq:aoi_regret_def} follows from the expectation of a geometric random variable with parameter $\mu^{*}$. 
\end{definition}
\color{black}

\section{Main Results and Discussion}
\label{sec:mainResults}
In this section we present our main contributions and their implications. First, we provide a lower bound on AoI regret for a certain class of policies, then we examine an upper bound on AoI regret for two policies, namely CUCB and CTS (proposed in \cite{gupta2019multiarmed}). Lastly, we derive the conditions under which the upper and lower bounds on AoI regret are order-wise equal.

\subsection{Lower Bound for Correlated Bandit Instances}
First, we define the class of $\alpha$-consistent policies and then in Theorem \ref{thm:lower_bound} provide a lower bound on the AoI regret achievable by any policy $\rho$ belonging to this class.
\begin{definition}[$\alpha$-consistent policies \cite{lai1985asymptotically}]
	Let $k_s$ denote the index of the channel scheduled in time-slot $s$. The index $k^{*}$ denotes the index of the optimal channel. A scheduling policy is called $\alpha$-consistent, for a constant $\alpha \in (0,1)$, if there exists an instance dependent constant $M$ such that,
	\begin{align}
		\mathbb{E} \Big[ \sum_{s=1}^{t} \mathds{1}\{k_s = k\} \Big] \leq M t^{\alpha}, \, \, \forall \, k \neq k^{*}.
	\end{align} 
	\label{def:alpha_consistent}
\end{definition}

\begin{theorem}[Lower bound on AoI regret]
	If a bandit instance $\mathrm{I}$ has at least one competitive arm $k$ with $\tilde{\Delta}_{k, k^{*}} < 0$, then for any $\alpha$-consistent policy $\rho$, we have,
	\begin{align*}
	R_{\rho}(T) &\geq  
	\max_{k \in \mathcal{C'}} \frac{\Delta_k}{\mathrm{D}(P_k, P'_k)} \frac{ (1 - \alpha)\log{T} - \log{(4M)} }{\mu^{*}}.
	\end{align*}
	Otherwise, if $\tilde{\Delta}_{k, k^{*}} \geq 0 \, \, \forall \, k \in [K]$,
$
	R_{\rho}(T) \geq  0.
$
	\label{thm:lower_bound}
\end{theorem}
Here $\mathrm{D}(P_k, P'_k)$ is the KL divergence between the reward distribution of arm $k$ and a suitably chosen perturbed reward distribution. The set $\mathcal{C'}$ is the set of strictly competitive arms and $M$ is an instance dependent constant as in Definition \ref{def:alpha_consistent}.

\subsection{Upper Bound for Correlated Bandit Instances}
We now characterize upper bounds on AoI regret for the policies CUCB and CTS proposed in \cite{gupta2019multiarmed}. The key idea behind these policies is to exploit the correlation in the rewards of various arms to obtain a set of competitive arms in each step. The algorithms then play an arm from this set.
For the sake of completeness, we provide details of these policies in the appendix (Algorithms \ref{algo:CUCB} and \ref{algo:CTS} respectively). 
\begin{theorem}[Upper bound on AoI regret under CUCB]
For a Correlated Bandit instance, let the number of competitive sub-optimal arms be equal to $C$. Further, let $\mathcal{C}$ denote the set of competitive arms inclusive of the optimal arm $k^{*}$. Let $\mu_{\min} = \min_k \mu_k$, 
\begin{align*}
	t_o &= \inf \Big \{ \tau \geq 2: \Delta_{\min}, \tilde{\Delta}_{k, k^{*}} \geq 4\sqrt{\frac{2K \log \tau }{\tau}} \Big \},\\
    U^{(nc)}_{k, \scalebox{.6}{$\mathrm{CUCB}$}} &= Kt_0 + K^{3}\sum_{t = Kt_0}^{T} 2\Big( \frac{t}{K} \Big)^{-2} + \sum_{t = 1}^{T}3t^{-3},\\
    U^{(c)}_{k, \scalebox{.6}{$\mathrm{CUCB}$}} &= 8 \frac{\log (T)}{\Delta_k^{2}} + \Big (1 + \frac{\pi^2}{3}\Big)
	+ \sum_{t=1}^{T} 2 K t \exp \Big( -\frac{t \Delta_{\min}^{2}}{2K} \Big).
\end{align*}
Then, for $T > t_0$,
\begin{align*}
\mathbb{E} [R_{\scalebox{.6}{$\mathrm{CUCB}$}}(T)] &\leq \frac{1 - \mu^{*}}{\mu^* \mu_{\min}} + \Big ( \frac{1}{\mu_{\min}} - \frac{1}{\mu^{*}}\Big ) \times\\
&\, \,  \bigg( \sum_{k^{'} \in [K] \backslash \mathcal{C}} \Delta_{k^{'}} U^{(nc)}_{k, \scalebox{.6}{$\mathrm{CUCB}$}} + \sum_{k \in \mathcal{C} \backslash \{k^{*}\}} \Delta_k U^{(c)}_{k, \scalebox{.6}{$\mathrm{CUCB}$}} \bigg)\\
&= \OO(1) + \OO(C \log T),
\end{align*}
and for $T \leq t_0$,
$
\mathbb{E} [R_{\scalebox{.6}{$\mathrm{CUCB}$}}(T)] \leq \big( \frac{1}{\mu_{\min}} - \frac{1}{\mu^{*}} \big) T.
$
\label{thm:upper_bound_cucb}
\end{theorem}

\begin{theorem}[Upper bound on AoI regret under CTS]
For a Correlated Bandit instance, let the number of competitive sub-optimal arms be equal to $C$ and let $\mathcal{C}$ denote the set of competitive arms inclusive of the optimal arm $k^{*}$. Further, let $\mu_{\min} = \min_k \mu_k$,
\begin{align*}
	t_b &= \inf \Big \{ \tau \geq \exp{(11 \beta)}: \Delta_{\min}, \tilde{\Delta}_{k, k^{*}} \geq 6\sqrt{\frac{2K \beta \log \tau }{\tau}} \Big \},\\
    U^{(nc)}_{k, \scalebox{.6}{$\mathrm{CTS}$}} &= K t_b + \sum_{t=1}^{T} 3 t^{-3} \\
	& \ \ + K^{2}\sum_{t = K t_b}^{T} \bigg( (2K + 3) \Big(\frac{t}{K}\Big)^{-2} + \Big( \frac{t}{K}\Big)^{1 - 2\beta}\bigg),\\
    U^{(c)}_{k, \scalebox{.6}{$\mathrm{CTS}$}} &= 18 \frac{\log (T \Delta_{k}^{2})}{\Delta_k^{2}} + \exp{(11\beta)} + \frac{9}{\Delta_{k}^{2}}\\
	& \ \ + \sum_{t=1}^{T} 2 K t \exp \Big( -\frac{t \Delta_{\min}^{2}}{2K} \Big).
\end{align*}
Then, for any choice of $\beta > 1$ and for $T > t_b$,
\begin{align*}
\mathbb{E} [R_{\scalebox{.6}{$\mathrm{CTS}$}}(T)] &\leq \frac{1 - \mu^{*}}{\mu^* \mu_{\min}} + \Big ( \frac{1}{\mu_{\min}} - \frac{1}{\mu^{*}}\Big ) \times\\
\, \, & \bigg( \sum_{k^{'} \in [K] \backslash \mathcal{C}} \Delta_{k^{'}} U^{(nc)}_{k, \scalebox{.6}{$\mathrm{CTS}$}} + \sum_{k \in \mathcal{C} \backslash \{k^{*}\}} \Delta_k U^{(c)}_{k, \scalebox{.6}{$\mathrm{CTS}$}} \bigg)\\
&= \OO(1) + \OO(C \log T),
\end{align*}
and for $T \leq t_b$,
$
\mathbb{E} [R_{\scalebox{.6}{$\mathrm{CTS}$}}(T)] \leq \big( \frac{1}{\mu_{\min}} - \frac{1}{\mu^{*}} \big) T.
$
\label{thm:upper_bound_cts}
\end{theorem}

\subsection{Discussion of Implications}
From Theorems \ref{thm:upper_bound_cucb} and \ref{thm:upper_bound_cts} it is clear that both CUCB and CTS are $\alpha$-consistent and therefore AoI regret under these policies will satisfy Theorem \ref{thm:lower_bound}. Since, the bounds on AoI regret depend on the number of sub-optimal competitive arms $C$, we consider different possibilities for $C$ to understand the regret bounds. If $C > 0$ for a Correlated Bandit instance, and if at least one arm is strictly competitive, then the lower bound and upper bound on AoI regret are both $\OO(\log T)$. However, when there are no competitive arms, that is when $C=0$, then there is no meaningful lower bound on the expected AoI regret. The $C = 0$ case agrees with the fact that the set $\mathcal{C}\backslash\{k^{*}\}$ being empty results in a constant $\OO(1)$ upper bound on AoI regret. Hence for both these cases of Correlated Bandit instances, the bounds are order-optimal.

The Correlated Bandit model used in this work assumes the knowledge of deterministic reward functions. In practice, it may be challenging to determine these functions, and while conducting such an exercise we may even be able to learn the distribution of the underlying state $X$ itself, thereby doing away with the need for any Bandit algorithm. However, the strength of this model is that once these functions are determined, they can be utilized in other communication systems with a similar configuration but a different and unknown distribution of $X$. Occlusions of different nature can repeatedly disrupt multiple channels at the same time engendering correlation in their performances. Communication systems similar to the one considered in this work might be placed in very different environments while following a standardized configuration. Due to differences in the environment, occlusion events analogous to the abstract entries in $\{x_1, x_2, \ldots, x_n\}$ used in this work, would occur with different relative frequencies at every installation. This variety in environments would make it difficult to scale solutions customized for every location. The distribution agnostic model and algorithms considered in this work would be highly beneficial in such scenarios.

\section{Simulations}
\begin{figure}[t]
	\centerline{\includegraphics[scale=0.41]{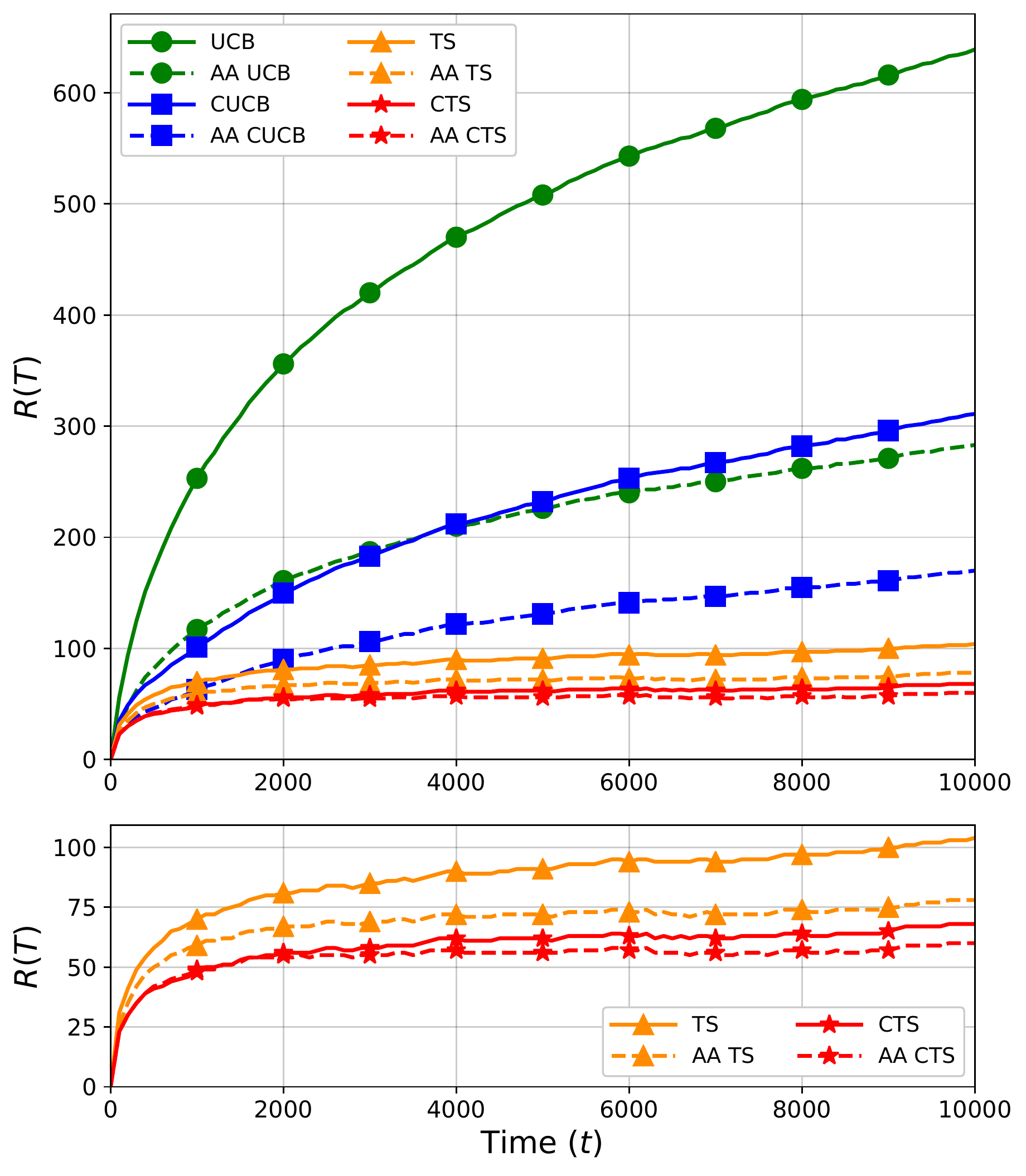}}
	\caption{AoI regret results for bandit instance I-1}
	\label{fig:i1_regret}
\end{figure}
\begin{figure}[t]
	\centerline{\includegraphics[scale=0.41]{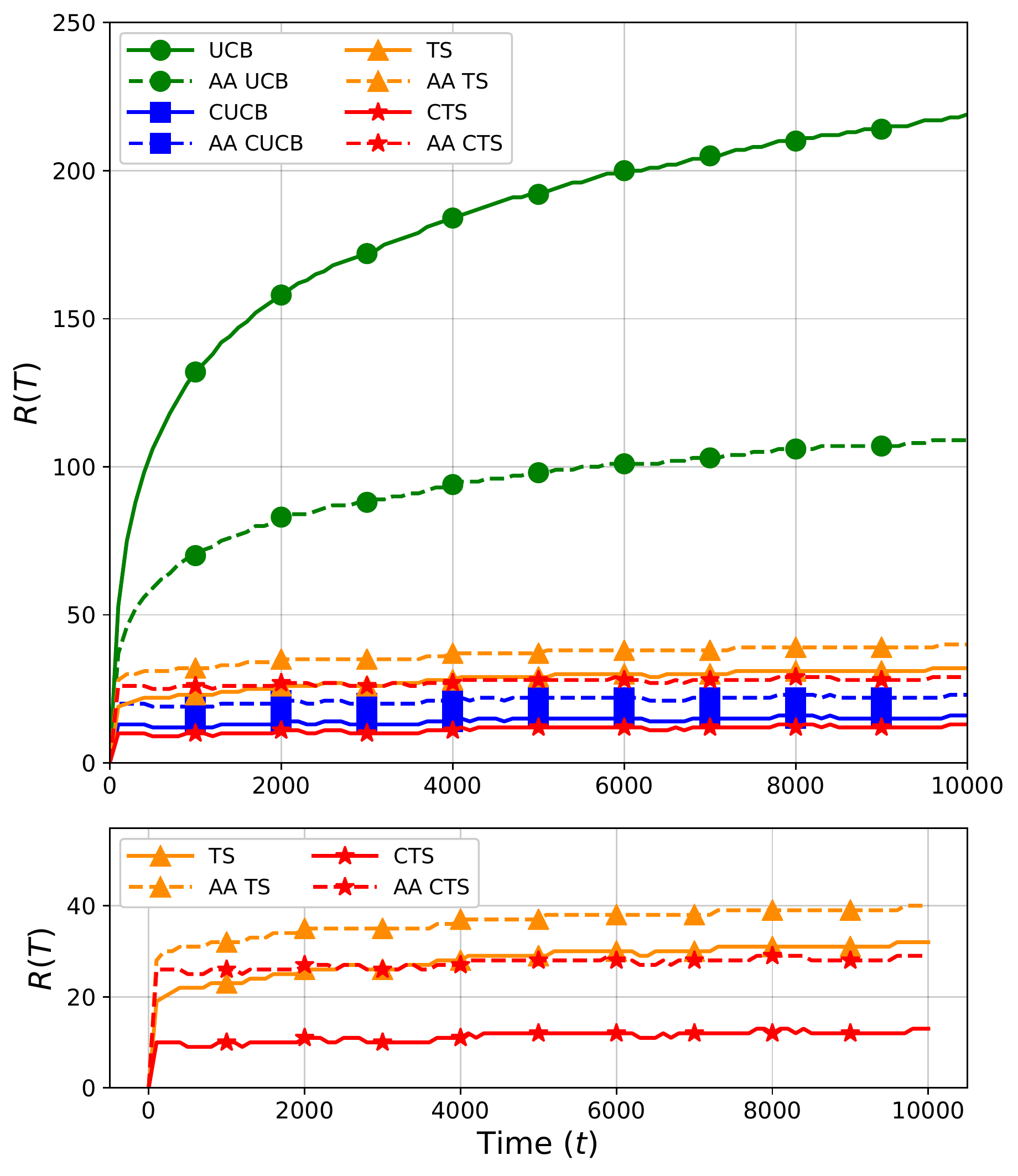}}
	\caption{AoI regret results for bandit instance I-2}
	\label{fig:i2_regret}
\end{figure}

In this section, we compare the performance of UCB, TS, CUCB, and CTS via simulations. In addition, we also compare the performance of \textit{AoI-aware} variants of these policies \cite{fatale2020regret}.

Whenever current AoI $a(t)$ is low, the AoI-aware variant of a policy makes decisions identical to its parent algorithm. However, when $a(t)$ exceeds a certain threshold specified in \cite{fatale2020regret}, the AoI-aware variant  selects the arm $k$ with the highest empirical mean $\hat{\mu}_k$. 

We plot AoI regret for the bandit instances $\mathrm{I}_1$ and $\mathrm{I}_2$, described in Section \ref{sec:setting} in Figures \ref{fig:i1_regret} and \ref{fig:i2_regret} respectively. It can be observed that for both $\mathrm{I}_1$ and $\mathrm{I}_2$, CTS and its AoI-aware variant perform the best on AoI regret. Further, CUCB and CTS have significantly lower AoI regret as compared to UCB and TS respectively. This is because of the fact that CUCB and CTS exploit correlation between arms to avoid sampling some non-competitive sub-optimal arms numerous times. Even though TS does not utilize correlation, its performance is still better than CUCB for $\mathrm{I}_1$. This is because, even though both UCB and TS have $\log T$ AoI regret scaling, the pre-factor for $\log T$ in UCB, and consequently in CUCB, is much larger than the prefactor in TS.

As seen in Section \ref{sec:setting} instance $\mathrm{I}_2$ had no competitive sub-optimal arms. Hence, as predicted by the bounds in Theorems \ref{thm:upper_bound_cucb} and \ref{thm:upper_bound_cts}, both CUCB and CTS have constant expected AoI regret regardless of the horizon. Interestingly, through Figure \ref{fig:i2_regret}, we also observe that the AoI-aware variant of a policy need not perform better than its parent policy as is the case for CUCB, CTS and TS. The regret bounds stated in Theorem \ref{thm:upper_bound_cts} were for Thompson Sampling with Gaussian priors, however, the simulation results for TS and CTS in this section were with Beta priors since the latter performs better on AoI regret.

\section{Proofs}
\label{sec:proofs}
\subsection{Proof of Theorem \ref{thm:lower_bound}}
\label{subsec:lower_lemmas}
\begin{lemma}[Lower bound on AoI regret for any policy]
	Let $n_k(t)$ be the number of times arm $k$ was scheduled in the time slots $1$ to $t-1$. Then, AoI regret $R_{\rho}(T)$ under any policy $\rho$ is lower bounded as,
	\begin{align*}
	R_{\rho}(T) \geq \frac{1}{\mu^{*}} \sum_{i \neq k^{*}} \Delta_{i} \mathbb{E}[n_i (T + 1)],
	\end{align*}
	where $\Delta_{i} = \mu^{*} - \mu_{i}$.
	\label{lemma:aoi_lb_general}
\end{lemma}
\begin{proof}
	Let $S_{\rho}(t)$ and $S^{*}(t)$ be indicator random variables denoting a successful update in time slot $t$ by policy $\rho$ and the oracle's policy respectively. Then, from Lemma 1 of \cite{fatale2020regret} we have,
	\begin{align}
		R_{\rho}(T) \geq \frac{1}{\mu^{*}} \sum_{t=1}^{T} \mathbb{E} [S^{*}(t) - S_{\rho}(t)].
		\label{eq:lb_statement}
	\end{align}
	Further, from \cite{fatale2020regret}, we also have,
	\begin{align}
		\mathbb{E} [S^{*}(t) - S_{\rho}(t)] &\geq \sum_{i \neq k^{*}} (\mu^{*} - \mu_k) \mathbb{P}(\mathds{1}\{ k_t = i \} = 1)\\
		&= \sum_{i \neq k^{*}} \Delta_{i} \mathbb{P} (\mathds{1}\{k_t = i\} = 1).
	\end{align}
	Therefore from inequality (\ref{eq:lb_statement}) we have,
	\begin{align}
	R_{\rho}(T) &\geq \frac{1}{\mu^{*}}\sum_{t=1}^{T} \sum_{i \neq k^{*}} \Delta_i \mathbb{P}(\mathds{1}\{ k_t = i \} = 1)\\
	&= \frac{1}{\mu^{*}} \sum_{i \neq k^{*}} \Delta_{i} \mathbb{E} [n_i (T + 1)].
	\end{align}   
\end{proof}

\begin{lemma}[Bretagnolle-Huber Inequality \cite{DBLP:Tsybakov}]\label{lemma:Bretagnolle_Huber}
	Consider two probability measures $P$ and $Q$, both absolutely continuous with respect to a given measure. Then for any event $\mathcal{A}$,
	\begin{align*}
	P(\mathcal{A}) + Q(\mathcal{A}^{c}) \geq \frac{1}{2} \exp\big( -\mathrm{KL}(P||Q) \big).
	\end{align*}
\end{lemma}

\begin{lemma}[Divergence Decomposition Lemma \cite{lattimore}]
	Let $\nu = (P_1 , \ldots, P_K)$ be the reward distributions associated with one $K$-armed bandit, and let $\nu' = ( P'_1 , \ldots , P'_K)$ be the reward distributions associated with another $K$-armed bandit. Let $\mathbb{P}_{\nu}^{t} = \mathbb{P}_{\nu \rho}^{t}$ and $\mathbb{P}_{\nu'}^{t} = \mathbb{P}_{\nu' \rho}^{t}$ be the joint distributions corresponding to the schedule of bandit arms chosen under policy $\rho$ and the rewards received. Then,
\begin{align*}
	\mathrm{KL}(\PP_{\nu}^{t}||\PP_{\nu'}^{t}) = \mathrm{D} (\PP_{\nu}^{t}, \PP_{\nu'}^{t}) = \sum_{i=1}^{K} \mathbb{E}_{\nu} [n_i (t + 1)] \mathrm{D}(P_{i}, P'_{i}),   
\end{align*}
	where $\mathbb{E}_{\nu} [n_i (t + 1)]$ is the expected number of pulls of arm $i$ in $t$ rounds of play for the bandit instance described by $\nu$.
	\label{lemma:divDecomposition}
\end{lemma}

\begin{proof}[Proof of Theorem \ref{thm:lower_bound}]
Consider a Correlated Bandit instance $\mathrm{I}$ having just two arms, the optimal arm with index $1$ and a lone sub-optimal arm with index $2$. If the sub-optimal arm is strictly competitive, that is if $\tilde{\Delta}_{2, 1} < 0$, then from Theorem 3 in \cite{gupta2019multiarmed} we can construct a perturbed bandit instance $\mathrm{I'}$ such that $\mathbb{E}_{X'}[\tilde{Y}_{2}(X)] > \mu_{1}$. Let $\mathbb{P}^{t}_{\mathrm{I}}$ and $\mathbb{P}^{t}_{\mathrm{I'}}$ be the distributions corresponding to rewards and scheduled arms in the first $t$ time steps for instances $\mathrm{I}$ and $\mathrm{I'}$ respectively. By construction, arm $2$ will be the optimal arm for the bandit instance $\mathrm{I'}$. Now, using Definition \ref{def:alpha_consistent}, for $\alpha$-consistent policies $\rho$, there exists a constant $M$ such that,
\begin{align}
\mathbb{E}^{t}_{\mathrm{I}} \Big[ \sum_{\tau = 1}^{t} \mathds{1} \{k_{\tau} = 2\} \Big] &\leq M t^{\alpha} \label{eq:markov1}\\
\mathbb{E}^{t}_{\mathrm{I'}} \Big[ \sum_{\tau = 1}^{t} \mathds{1} \{k_{\tau} = 1\} \Big] &\leq M t^{\alpha}. 
\label{eq:markov2}
\end{align}
Defining the event $\mathcal{A} = \{n_2(t + 1) > t/2\}$ and using inequalities (\ref{eq:markov1}) and (\ref{eq:markov2}), the following Markov inequalities hold,
\begin{align}
\mathbb{P}^{t}_{\mathrm{I}} (\mathcal{A}) &\leq \frac{2M}{t^{1 - \alpha}}\\
\mathbb{P}^{t}_{\mathrm{I'}} (\mathcal{A}^{c}) &\leq \frac{2M}{t^{1 - \alpha}}.
\end{align}
Now, using Lemma \ref{lemma:Bretagnolle_Huber} we can write,
\begin{align}
\mathrm{D} (\mathbb{P}_{\mathrm{I}}^{t}, \mathbb{P}_{\mathrm{I'}}^{t}) \geq (1 - \alpha)\log{t} - \log{(4M)}.
\label{eq:div_lower_bound}
\end{align}
Next, using Lemma \ref{lemma:divDecomposition} we can expand $\mathrm{D} (\mathbb{P}_{\mathrm{I}}^{t}, \mathbb{P}_{\mathrm{I'}}^{t})$ as,
\begin{align}
\mathrm{D} (\mathbb{P}_{\mathrm{I}}^{t}, \mathbb{P}_{\mathrm{I'}}^{t}) &=  \mathbb{E}_{\mathrm{I}} [n_2(t+1)] \mathrm{D} (P_2, P'_2).
\label{eq:lb_on_pulls}
\end{align}
Combining inequality (\ref{eq:lb_on_pulls}) and Lemma \ref{lemma:aoi_lb_general}, we get the following lower bound on AoI regret for instance $\mathrm{I}$,
\begin{align}
R_{\rho}(T) &\geq \frac{\Delta_2 \Big( (1 - \alpha)\log{T} - \log{(4M)} \Big)}{\mu^{*}{\mathrm{D} (P_2, P'_2)} }. \label{eq:lower_bound_I}
\end{align}
If a Correlated Bandit instance has more than one strictly competitive sub-optimal arm, then the expected number of sub-optimal pulls, and by extension, the lower bound on AoI regret, can only be higher due to the greater exploration required. Hence, in general, if the collection of strictly competitive arms is given by $\mathcal{C'}$, then, 
\begin{align}
R_{\rho}(T) &\geq \max_{k \in \mathcal{C'}} \frac{\Delta_k}{\mathrm{D}(P_k, P'_k)} \frac{ (1 - \alpha)\log{T} - \log{(4M)} }{\mu^{*}}. \label{eq:lower_bound}
\end{align}
If among the sub-optimal arms, there is no strictly competitive arm, then the lower bound on AoI Regret is simply $0$.
\end{proof}

\subsection{Proof of Theorems \ref{thm:upper_bound_cucb} and \ref{thm:upper_bound_cts}}
As in ordinary MAB instances, the reward $Y_k(X)$ of each arm $k$ in Correlated Bandit instances is distributed according to the Bernoulli distribution. Hence, under the technical Assumption 1 of \cite{fatale2020regret} the following Lemma would apply to Correlated Bandit instances.
\begin{lemma}[Lemma 4 in \cite{fatale2020regret}]\label{lemma:gen_upper_bound}
	Given the expected number of pulls of sub-optimal arms, the AoI regret over a run of $T$ rounds is upper bounded by,
	\begin{align*}
	\sum_{t=1}^{T} \mathbb{E} [a(t)] - \frac{T}{\mu^{*}} \leq 
	\frac{1 - \mu^{*}}{\mu^* \mu_{\min}} + \Big ( \frac{1}{\mu_{\min}} - \frac{1}{\mu^{*}}\Big ) \E \Big[\sum_{t=1}^{T} \mathds{1}_{k_t \neq k^{*}}\Big].
	\label{UB_general}
	\end{align*}
	\label{lemma:aoi_upper_bound}
\end{lemma}

\begin{lemma}[Expected number of pulls of a non-competitive arm, Theorem 1 in \cite{gupta2019multiarmed}]
	The expected number of pulls of a non-competitive sub-optimal arm under CUCB is upper bounded by,
	\begin{align*}
	\mathbb{E} [n_{k}^{(nc)}(T)]  &\leq Kt_0 + K^{3}\sum_{t = Kt_0}^{T} 2\Big( \frac{t}{K} \Big)^{-2} + \sum_{t = 1}^{T}3t^{-3}\\
	&= U^{(nc)}_{k, \scalebox{.6}{$\mathrm{CUCB}$}} = O(1),
	\intertext{and under CTS with Gaussian priors by,}
	\begin{split}
	\mathbb{E} [n_{k}^{(nc)}(T)] &\leq K t_b + \sum_{t=1}^{T} 3 t^{-3} \, \\
	& \quad \, + K^{2}\sum_{t = K t_b}^{T} \bigg( (2K + 3) \Big(\frac{t}{K}\Big)^{-2} + \Big( \frac{t}{K}\Big)^{1 - 2\beta}\bigg)
	\end{split}\\
	&= U^{(nc)}_{k, \scalebox{.6}{$\mathrm{CTS}$}} = O(1),
	\end{align*}
	where $t_0, t_b > 0$ are constants defined as,
	\begin{align*}
	t_o &= \inf \Big \{ \tau \geq 2: \Delta_{\min}, \tilde{\Delta}_{k, k^{*}} \geq 4\sqrt{\frac{2K \log \tau }{\tau}} \Big \},\\
	t_b &= \inf \Big \{ \tau \geq \exp{(11 \beta)}: \Delta_{\min}, \tilde{\Delta}_{k, k^{*}} \geq 6\sqrt{\frac{2K \beta \log \tau }{\tau}} \Big \},    
	\end{align*}
	where $\beta > 1$ is a parameter in CTS (Gaussian priors). 
	\label{lemma:non_comp}
\end{lemma}

\begin{lemma}[Expected number of pulls of a competitive arm, Theorem 2 in \cite{gupta2019multiarmed}]
	The expected number of pulls of a competitive sub-optimal arm under CUCB is upper bounded by, 
	\begin{align*}
	\begin{split}
	\mathbb{E} [n_{k}^{(c)}(T)] &\leq 8 \frac{\log (T)}{\Delta_k^{2}} + \Big (1 + \frac{\pi^2}{3}\Big)\\
	&\qquad \, \,  + \sum_{t=1}^{T} 2 K t \exp \Big( -\frac{t \Delta_{\min}^{2}}{2K} \Big)
	\end{split}\\
	&= U^{(c)}_{k, \scalebox{.6}{$\mathrm{CUCB}$}} =  O(\log T),
	\intertext{and under CTS with Gaussian priors by,}
	\begin{split}
	\mathbb{E} [n_{k}^{(c)}(T)] &\leq 18 \frac{\log (T \Delta_{k}^{2})}{\Delta_k^{2}} + \exp{(11\beta)} + \frac{9}{\Delta_{k}^{2}}\\
	& \qquad \, \, + \sum_{t=1}^{T} 2 K t \exp \Big( -\frac{t \Delta_{\min}^{2}}{2K} \Big)
	\end{split}\\
	&= U^{(c)}_{k, \scalebox{.6}{$\mathrm{CTS}$}} = O(\log T),
	\end{align*}
	where $\beta > 1$ is a parameter in CTS (Gaussian priors).
	\label{lemma:comp}
\end{lemma}

\begin{proof}[Proof of Theorems \ref{thm:upper_bound_cucb} and \ref{thm:upper_bound_cts}]
	The results follow by substituting the appropriate expression for the total number of sub-optimal pulls from Lemmas \ref{lemma:non_comp} and \ref{lemma:comp} into Lemma \ref{lemma:gen_upper_bound}.
\end{proof}

\section{Conclusions and Future Work}
%
With next generation communication technologies relying on higher carrier frequencies, the problem of line-of-sight occlusion becomes more significant. In systems of the kind considered in this work, occlusions can effect multiple channels simultaneously resulting in their performances being correlated. This correlation can be exploited to identify certain channels as sub-optimal within a few time steps. We have shown theoretical bounds on AoI regret to corroborate this fact. Moreover, through simulations we observed that policies capable of exploiting correlation perform significantly better than those that disregard the presence of correlation.

The Correlated Bandit framework used in this work does not put constraints on the nature of rewards received. It would likely be beneficial to exploit the fact that in our system, rewards are always either $0$ or $1$. Additionally, the theoretical upper bound on AoI regret for CTS was for \textit{Gaussian} priors. An analogous result remains to be shown for CTS with \textit{Beta} priors. These topics will be the focus of future work.

\bibliographystyle{IEEEtran}
\bibliography{IEEEabrv,ref}

\newpage
\appendix
\begin{algorithm}[ht]
	\DontPrintSemicolon
	\textbf{Input:} Pseudo-rewards $s_{\ell, k}(r)$ \;
	\textbf{Initialize:} Set $\hat{\mu}_k$, $\hat{\phi}_{\ell, k}$ and $n_k$ as $0$ $\forall$ $k\in[K]$. \;
	\While{$1\leq t \leq K$}{
		Schedule update on Channel $k_t=t$ \;
		Receive reward $r_t$ drawn from $\text{Ber}(\mu_{k_t})$ \;
		$\hat{\mu}_{k_t} = r_t$\;
		$n_{k_t}(t) = 1$ \;
		$t = t+1$}
	\While{$ t \geq K + 1 $}{
		Find $\mathcal{S}_t = \{k: n_k(t-1) \geq \frac{t-1}{K} \}$, the set of arms pulled a significant number of times till $t-1$. Define $k^{\mathrm{emp}}(t) = \arg \max_{k \in \mathcal{S}_t} \hat{\mu}_{k_t}$ \;
		Initialize the empirically competitive set $\mathcal{A}_t$ as $\{\}$ \;
		\For{$k \in [K]$}{
			\If{$\min_{\ell \in S_t} \hat{\phi}_{k, \ell}(t) \geq \hat{\mu}_{k^{\mathrm{emp}}}(t)$}{Add empirically competitive arms $k$ to the set: $\mathcal{A}_t = \mathcal{A}_t \cup \{ k \}$}
		} 
		Schedule update on Channel $k_t$ such that, $ k_t=\arg \max_{k\in \mathcal{A}_t \cup \{ k^{\mathrm{emp}}(t)\}} \, \hat{\mu}_{k_t}+\sqrt{\frac{2\log t}{n_k(t-1)}} $ \;
		Receive reward $r_t$ drawn from $\text{Ber}(\mu_{k_t})$ \;
		$\hat{\mu}_{k_t}=(\hat{\mu}_{k_t}\cdot n_{k_t}(t-1) + r_t)/(n_{k_t}(t-1) + 1)$ \;
		$n_{k_t}(t) = n_{k_t}(t-1) + 1$ \;
		$\hat{\phi}_{k, k_t}= \sum_{\tau:k_{\tau}=k_t} s_{k, k_{\tau}} (r_{\tau}) / n_{k_t}(t) \, \, \forall \, k \neq k_t$ \;
		$t = t + 1$}
	\caption{\sc Correlated UCB (CUCB)}
	\label{algo:CUCB}
\end{algorithm}
\begin{algorithm}[ht]
	\DontPrintSemicolon
	\textbf{Input:} Pseudo-rewards $s_{\ell, k}(r)$ \;
	\textbf{Initialize:} Set the number of successes $S_k(t)$, failures $F_k(t)$, and the quantities in CUCB as $0$ $\forall$ $k\in[K]$. \;
	\While{$ t \geq 1 $}{
		Perform the steps 10 - 14 as in Algorithm \ref{algo:CUCB} \;
		For each $k$ in $[K]$, draw a sample $\theta_k(t)$, where,
		$ \theta_k(t) \sim \text{Beta}(S_k(t-1) + 1, F_k(t-1) + 1) $ \;
		Schedule update on Channel $k_t$ such that 
		$ k_t = \arg \max_{k\in \mathcal{A}_t \cup \{ k^{\text{emp}}(t)\}} \theta_k(t)$ \;
		Receive reward $r_t$ drawn from $\text{Ber}(\mu_{k_t})$ \;
		$S_{k_t}(t) = S_{k_t}(t-1) + r_t$ \;
		$F_{k_t}(t) = F_{k_t}(t-1) + (1 - r_t)$ \;
		Perform the steps 17 - 20 as in Algorithm \ref{algo:CUCB} \;
	}
	\caption{\sc Correlated TS (CTS)}
	\label{algo:CTS}
\end{algorithm}

\end{document}